\documentclass{amsart}

\usepackage{color}
\usepackage{amssymb, amsmath}
\usepackage{amsthm} 
\usepackage{hyperref}

\usepackage[foot]{amsaddr}    

\newtheorem{theorem}{Theorem}
\newtheorem{proposition}[theorem]{Proposition}
\newtheorem{definition}{Definition}
\newtheorem{lemma}[theorem]{Lemma}

\def\IMARK{\mbox{\emph{i}-{\sc Mark}}}
\def\MARK{\mbox{\sc Mark}}
\def\IMIMARK{\mbox{\emph{i}-{\sc MiMark}}}
\def\MIMARK{\mbox{\sc MiMark}}
\def\UPMARK{\mbox{\sc UpMark}}
\def\MARKT{\mbox{{\sc Mark}-\emph{t}}}

\def\NIM{\mbox{\sc Nim}}
\def\NONAME{\mbox{\sc Game}}
\def\GAME{\mbox{\sc Game}}

\newcommand\ESOK[1]{#1}

\def\NNN{\mathbb{N}}

\def\Bn{\mathbf{n}}
\def\Bk{\mathbf{k}}
\newcommand\B[1]{\mathbf{#1}}

\def\mex{{\rm mex}}
\def\opt{{\rm opt}}
\def\xor{\ {\rm XOR}\ }

\def\NN{\mathcal{N}}
\def\PP{\mathcal{P}}
\def\GG{\mathcal{G}}
\def\EE{\mathcal{E}}

\def\ESalpha{q}
\def\ESbeta{p}

\author[E. Sopena]{\'Eric Sopena}
\email{eric.sopena@labri.fr}

\address{Univ. Bordeaux, LaBRI, UMR5800, F-33400 Talence, France.}
\address{CNRS, LaBRI, UMR5800, F-33400 Talence, France.}

\thanks{\ESOK{This work has been supported by the ANR-14-CE25-0006 project of the French National Research Agency.}}

\title{\IMARK: A New Subtraction Division Game}

\date{\today}

\begin{document}


\begin{abstract}
Given two finite sets of integers $S\subseteq\NNN\setminus\{0\}$ and $D\subseteq\NNN\setminus\{0,1\}$,
the impartial combinatorial game $\IMARK(S,D)$ is played on a heap of tokens. 
From a heap of $n$ tokens, each player can move
either to a heap of $n-s$ tokens for some $s\in S$, or to a heap of $n/d$ tokens
for some $d\in D$ if $d$ divides $n$.
Such games can be considered as an integral variant of \MARK-type games, introduced by Elwyn Berlekamp and Joe Buhler
and studied by Aviezri Fraenkel and Alan Guo, for which it is allowed to move from a heap of $n$ tokens
to a heap of $\lfloor n/d\rfloor$ tokens for any $d\in D$.

Under normal convention, it is observed that the Sprague-Grundy sequence of the game $\IMARK(S,D)$ is aperiodic for any sets $S$ and $D$.
However, we prove that, in many cases, this sequence is almost periodic and that the set of winning positions is periodic.
Moreover, in all these cases, the Sprague-Grundy value of a heap of $n$ tokens can be computed in time $O(\log n)$.

We also prove that, under mis\`ere convention, the outcome sequence of these games is purely periodic.
\end{abstract}

\maketitle

\medskip


\medskip 

\noindent
{\bf  Mathematics Subject Classification (MSC 2010):} 91A46.

\section{Introduction}
\label{sec:intro}

The impartial combinatorial game \MARK, due to Mark Krusemeyer according to Elwyn Berlekamp and Joe Buhler~\cite{BB09}, is played on a heap
of tokens. On their turn, each player can move from a heap of $n$ tokens, $n\ge 1$,
either to a heap of $n-1$ tokens or to a heap of $\lfloor n/2\rfloor$ tokens.
Under {\em normal convention}, the first player unable to move (when the heap is empty) loses the game.
This game is a particular case of a more general family of games, that we call 
{\em subtraction division games} following~\cite{K12}, defined as 
follows\,\footnote{In~\cite{K12}, Elizabeth Kupin used a slightly different definition
for subtraction division games, where division-type moves leave a heap of 
$\lceil n/d\rceil$ tokens, instead of $\lfloor n/d\rfloor$ tokens,
and the game
stops as soon as the heap contains a unique token, instead of no token at all.}.
Let us denote by $[a,b]$, $a\le b$, the set of integers $\{i\ |\ a\le i\le b\}$,
by $\NNN_{\ge i}$, $i\in\NNN$, the set $\NNN\setminus [0,i-1]$
and let $S\subseteq\NNN_{\ge 1}$ and $D\subseteq\NNN_{\ge 2}$ be two
finite sets of integers. The subtraction division game $SD(S,D)$ is played on a heap
of $n$ tokens, $n\ge 0$. From a heap of $n$ tokens, each player can move
either to a heap of $n-s$ tokens for some $s\in S$, or to a heap of $\lfloor n/d\rfloor$ tokens
for some $d\in D$.
The \MARK\ game is therefore the game $SD(\{1\},\{2\})$.

In this paper, we study a variant of such subtraction division games,
that we propose to name
 \IMARK$(S,D)$ for {\em integral} \MARK, obtained by restricting division-type moves, allowing a move
from a heap of $n$ tokens to a heap of $\lfloor n/d\rfloor$ tokens,
$d\in D$, only when $d$ divides $n$, so that $\lfloor n/d\rfloor=n/d$.

\subsection{Combinatorial Game Theory: basic notions and terminology}

We first briefly recall the main notions of Combinatorial Game Theory that we will
need in this paper. More details can be found on any of the reference books
 \cite{LIP}, \cite{WW}, \cite{ONAG} or~\cite{CGT}.

Let us call a {\em heap-game} any impartial combinatorial game played on a heap of tokens.
We denote by $\Bn$ a heap of $n$ tokens, or $\Bn_{\mbox{\small\sc Game}}$ if we want to specify that the considered game is \GAME.
Those 
heaps $\Bn$ for which the {\bf N}ext player to move has a winning strategy
are N-positions, whereas those for which the {\bf P}revious player
has a winning strategy are P-positions.
We denote by $\NN$ the set of P-positions and by $\PP$ the set of P-positions.
Note that a position $\Bn$ is in $\PP$ if and only if every move from $\Bn$ leads
to a position in $\NN$, whereas $\Bn$ is in $\NN$ if and only if there exists at least
one move from $\Bn$ to a position in $\PP$.
We also say that the {\em outcome} of a position $\Bn$ is N if $\Bn\in\NN$ and P otherwise.
It is not difficult to prove that any position of an impartial combinatorial
game is either an N-position or a P-position, which implies that the two
sets $\NN$ and $\PP$ are {\em complementary} -- that is, $\NN\cap\PP=\emptyset$
and $\NN\cup\PP=\NNN$.

Heap-games can be played on any finite number of heaps, by means of {\em sums of games}.
In that case, a player move consists in first selecting one of the heaps and then
making a legal move on that heap. Knowing whether each heap is an N- or a P-position is
not sufficient to determine whether the whole game is an N- or a P-position.
Under normal convention, we use the {\em Sprague-Grundy function} which assigns to
every position $G$ of an impartial game its {\em $g$-value} $g(G)$ (sometimes called {\em nim-value} or {\em nimber}), defined as the unique positive integer $k$ such that $G$
is equivalent to the \NIM\ position $\Bk_{\NIM}$.
(Equivalence here means that the sum  $G+\Bk_{\NIM}$ is a P-position.)
Therefore, a position is in $\PP$ if and only if its $g$-value is 0.
The $g$-value of any position $G$ of an impartial game can be inductively computed using the {\em mex function},
defined by 
$$\mex(S)=\min(\NNN\setminus S)$$
 for every finite set $S\subseteq\NNN$. We then have
$$g(G)=\mex(\opt(G)),$$
 where $\opt(G)$ denotes the set of {\em options} of $G$ -- that is, the
set of positions that can be reached from $G$ by making a legal move.
Note here that if a position has $k$ options then its $g$-value is at most $k$.

Let $\oplus$ denote the {\em nim-sum function}, defined by 
$$n_1\oplus n_2=R^{-1}(R(n_1)\xor R(n_2))$$
 for any two
positive integers $n_1$ and $n_2$, where $R$ denotes the function that associates with each integer its
binary representation. For a position $G=G_1+\ldots +G_p$ of a sum of games, 
the $g$-value of $G$ is then given by
$$g(G)=g(G_1)\oplus\ldots\oplus g(G_p).$$

Computing the $g$-value of a game recursively using the mex function can be exponential in the `size' of $G$ and thus
inefficient.
For a heap-game \NONAME\, we define the {\em Sprague-Grundy sequence} (sometimes called the {\em nim-sequence}) of \NONAME\
as the sequence 
$$\GG\mbox{(\NONAME)}=(g({\mathbf 1}_{\NONAME}),g({\mathbf 2}_{\NONAME}),\ldots).$$
Knowing this sequence clearly allows to compute the $g$-value of any position of \NONAME\ on a finite
number of heaps. This can be efficient whenever the Sprague-Grundy sequence has ``nice properties'',
in particular if it is periodic.
We say that an integer sequence $(a_i)_{i\in\NNN}$ is {\em periodic},
with preperiod of length $\ESalpha\ge 0$ and period of length $ \ESbeta\ge 1$, if for every $i\ge\ESalpha$,
$a_i=a_{i+ \ESbeta}$. In such a case, any value $a_i$ of the sequence is 
determined by the value of $a_i\mod \ESbeta$, which can be
computed in time $O(\log a_i)$ for a fixed $ \ESbeta$.
This is namely the case for all {\em subtraction games} -- that is, games of the form $SD(S,\emptyset)$ -- whenever
the subtraction set $S$ is finite~\cite{WW}.

Similarly, we  say that a set of integers $S\subseteq\NNN$ is periodic if there exist $\ESalpha$, $ \ESbeta$,
with $\ESalpha\ge 0$ and $ \ESbeta\ge 1$, such that for every $i\ge\ESalpha$, $i\in S$ if and only if $i+ \ESbeta\in S$.
Note that since the sets $\NN$ and $\PP$ of any heap-game are complementary, 
either both these sets are periodic or none of them is periodic.

We define the {\em outcome sequence} of a game \NONAME\
as the sequence $(o_i)_{i\in\NNN}$ given by $o_i=N$ if $\B{i}$ is an N-position for \NONAME\ and $o_i=P$ otherwise.
We will say that such an outcome sequence is periodic whenever the sets $\NN$ and $\PP$ are periodic.

\subsection{The game \MARK}

In~\cite{F11,F12}, Aviezri Fraenkel developped a study of 
the \MARK\ game $SD(\{1\},\{2\})$ and gave a characterization
of the corresponding sets $\NN$ and $\PP$. Although aperiodic, these two sets reveals a nice structure:
a \MARK\ position $\Bn$ is in $\NN$ if and only if the binary representation $R(n)$ of $n$ has an even number of 
trailing 0's\footnote{Aviezri Fraenkel called such numbers \emph{vile} numbers, and \emph{dopey} numbers those numbers
whose binary representation has an odd number of trailing 0's.}.
This allows to compute the outcome of any \MARK\ position $\Bn$ in time $O(\log n)$.
It is also proved that the Sprague-Grundy sequence $\GG$(\MARK) has the following property (note that since
any \MARK\ position has at most two options, $g(\Bn)\le 2$ for any \MARK\ position $\Bn$):
\begin{itemize}
\item $g(\Bn)=0$ is and only if $R(n)$ has an odd number of trailing 0's,
\item $g(\Bn)=1$ if and only if $R(n)$ has an even number of trailing 0's and an odd number of 1's,
\item $g(\Bn)=2$ if and only if $R(n)$ has an even number of trailing 0's and an even number of 1's.
\end{itemize}
Again, this allows to compute the $g$-value $g(\Bn)$ of any \MARK\ position $\Bn$ in time $O(\log n)$.

Aviezri Fraenkel also studied the \MARK\ game under {\em mis\`ere convention} (the first player unable to move wins the game),
the game \UPMARK\ (allowing moves from $\Bn$ to $\mathbf{\lceil n/2\rceil}$ instead of $\mathbf{\lfloor n/2\rfloor}$)
and introduced the general game \MARKT\ $=SD([1,t-1],\{t\})$, for any given $t\ge 1$, whose sets $\NN$ and $\PP$ again are aperiodic for every $t$.

This latter game has been studied by Alan Guo in~\cite{G12}. He proved that the $g$-value of any \MARKT\ position can be computed
in quadratic time. More precisely, he proved that $g(\Bn)=k$, $k\le t-2$, if and only if $R_t(n)$ has an odd number
of trailing $k$'s, where $R_t(n)$ denotes the representation of $n$ written in base $t$,
and that deciding whether $g(\Bn)=t-1$ or $g(\Bn)=t$ can be done in quadratic time.


\subsection{The game \IMARK: an integral subtraction division game}

We study in this paper an integral variant of \MARK-type games.
For any two sets $S\subseteq\NNN_{\ge 1}$ and $D\subseteq\NNN_{\ge 2}$,
the allowed moves in the game $\IMARK(S,D)$ are those leading from
a position $\Bn$ to any of the positions $\B{n-s}$, with $s\in S$,
or $\B{n/d}$, with $d\in D$ and $d|n$.

As for the game \MARK, the Sprague-Grundy sequence of the game $\IMARK(S,D)$
is aperiodic whenever the set $D$ is non-empty:

\begin{theorem}
For every finite set $S\subseteq\NNN_{\ge 1}$ and every non-empty finite set $D\subseteq\NNN_{\ge 2}$, the Sprague-Grundy sequence of the game i-$\MARK(S,D)$
is aperiodic.
\label{th:aperiodic}
\end{theorem}

\begin{proof} 
Assume to the contrary that $\GG(\IMARK(S,D))$ is periodic, with preperiod of length $\ESalpha$ and period of length $ \ESbeta$,
 and let $d\in D$.
Let $\Bn$ be any position with $n=k \ESbeta\ge\ESalpha$ for some $k\ge 1$.
Since $d\in D$, there is a move from $\mathbf{dn}$ to $\Bn$ and, therefore $g(\mathbf{dn})\neq g(\Bn)$.
This contradicts our assumption 
since $dn-n=(d-1)n=(d-1)k \ESbeta$.
\end{proof}

However, and in contrast to the game \MARK, we will prove that the outcome sequence
of such integral games is periodic in many cases.
Moreover, the Sprague-Grundy sequence appears to be ``almost periodic'' in most of  these cases,
in the following sense:

\begin{definition}{\rm
An integer sequence $(a_i)_{i\in\NNN}$ is {\em $\ell$-almost periodic} for some $\ell\ge 0$,
with preperiod of length $\ESalpha\ge 0$, period of length $ \ESbeta>\ell$ and {\em exception set}
$\EE=\{j_1,\ldots,j_\ell\}\subseteq [0, \ESbeta-1]$ if, for every $i\ge\ESalpha$ with $(i\mod \ESbeta)\notin\EE$,
$a_i=a_{i+ \ESbeta}$.
}
\label{def:almost-periodic}
\end{definition}

Intuitively speaking, a sequence is $\ell$-almost periodic if
it is ``periodic, except on $\ell$ columns of the period''.
A 0-almost periodic sequence is thus a periodic sequence.

\subsection{Organisation of the paper}

In the next section, we provide more details on the results proposed by Aviezri Fraenkel in~\cite{F11}
on the game \MARK\ and consider its integral version $\IMARK(\{1\},\{2\})$ as an introductory example.

We then study games of the form
$\IMARK([1,t-1],\{d\})$, with $t,d\ge 2$, in Section~\ref{sec:imark-t-d}, 
and of the form $\IMARK(\{a,2a\},\{2\})$, with $a\ge 2$, in Section~\ref{sec:imark-2a}.

We finally consider the mis\`ere version of the game \IMARK\ in Section~\ref{sec:misere}
and propose some open questions in Section~\ref{sec:discussion}.

\bigskip

\section{A first sample game: \IMARK$(\{1\},\{2\})$}
\label{sec:imark-12}

We consider in this section the simple game $\IMARK(\{1\},\{2\})$ since it corresponds
to our variant of the original game \MARK. We first briefly recall some results presented
by Aviezri Fraenkel in~\cite{F11}.

\subsection{The game \MARK}

Let $A=(a_n)_{n\ge 1}$ and $B=(b_n)_{n\ge 0}$ be the two sequences recursively defined by:

\medskip

\noindent\mbox{}\hskip 3cm
 $a_n = \mex\{a_i,b_i\ |\ 0\le i<n\},$

\medskip

\noindent\mbox{}\hskip 3cm
 $b_0=0 \mbox{ and } b_n=2a_n,\ n\ge 1.$

\medskip

These two sequences are clearly complementary and respectively correspond to the sequences A003159 and A036554 of the
``On-Line Encyclopedia of Integer Sequences'' of Neil Sloane~\cite{OEIS}. 

A position $\Bn$ is then an N-position for \MARK\ if and only if $n\in A$, and thus a P-position
for \MARK\ if and only if $n\in B$.
The first elements of these two sequences
are given in the following table:

\medskip

\begin{center}
  	\begin{tabular}{|c|ccccccccccccccccc|}
	\hline
	$n$ & 0 & 1 & 2 & 3 & 4 & 5 & 6 & 7 & 8 & 9 & 10 & 11 & 12 & 13 & 14 & 15 & 16 \\
	\hline 
	$a_n$ &  & 1 & 3 & 4 & 5 & 7 & 9 & 11 & 12 & 13 & 15 & 16 & 17 & 19 & 20 & 21 & 23\\
	$b_n$ & 0 & 2 & 6 & 8 & 10 & 14 & 18 & 22 & 24 & 26 & 30 & 32 & 34 & 38 & 40 & 42 & 46\\
	\hline
	\end{tabular}
\end{center} 

\medskip

Moreover, it can be observed that both these sets are aperiodic and, therefore,
the Sprague-Grundy sequence $\GG(\MARK)$ is aperiodic. The first elements of this sequence
are the following:

\medskip

\begin{center}
  	\begin{tabular}{|c|ccccccccccccccccc|}
	\hline
	$n$ & 0 & 1 & 2 & 3 & 4 & 5 & 6 & 7 & 8 & 9 & 10 & 11 & 12 & 13 & 14 & 15 & 16 \\
	\hline
	$g(\Bn)$ &  0 & 1 & 0 & 2 & 1 & 2 & 0 & 1 & 0 & 2 & 0 & 1 & 2 & 1 & 0 & 2 & 1 \\ 
	\hline
	\end{tabular}
\end{center}


\subsection{The game $i$-\MARK$(\{1\},\{2\})$}
\label{subsec:imark-12}

Let us now consider the integral version of the game \MARK, denoted $\IMARK(\{1\},\{2\})$. 
By Theorem~\ref{th:aperiodic}, we know that the Sprague-Grundy sequence
of the game $\IMARK(\{1\},\{2\})$ 
is aperiodic.
We first prove that its outcome sequence is periodic.

\begin{theorem}
The outcome sequence of the game i-$\MARK(\{1\},\{2\})$ is periodic, with preperiod of length~4
and period of length~2.
More precisely, we have
\begin{itemize}
\item[{\rm (i)}] $\NN=\{\B{1},\B{3}\}\cup\{\B{2k}\ |\ k\ge 2\}$,
\item[{\rm (ii)}] $\PP=\{\B{0},\B{2}\}\cup\{\B{2k+1}\ |\ k\ge 2\}$.
\end{itemize}
\label{th:NP-IMARK-1-2}
\end{theorem}

\begin{proof} 
Clearly, $\B{0}\in\PP$.
Since $\opt(\B{1})=\{\B{0}\}$, $\opt(\B{2})=\{\B{1}\}$, $\opt(\B{3})=\{\B{2}\}$ and $\opt(\B{4})=\{\B{2},\B{3}\}$,
we get $\B{1},\B{3},\B{4}\in\NN$ and $\B{2}\in\PP$.
The result then follows by induction since 
(i) $\opt(\Bn)=\{\B{n-1}\}$ for every odd $n$, so that $\Bn\in\PP$, and
(ii) for every P-position $\Bn$, $\Bn\in\opt(\B{n+1})$ and thus $\B{n+1}$ is an N-position.
\end{proof}

The first elements of the Sprague-Grundy sequence $\GG(\IMARK(\{1\},\{2\}))$ are the following:

\medskip

\begin{center}
  	\begin{tabular}{|c|cccccccccccccccc|}
	\hline
	$n$ & 0 & 1 & 2 & 3 & 4 & 5 & 6 & 7 & 8 & 9 & 10 & 11 & 12 & 13 & 14 & 15  \\
	\hline
	$g(\Bn)$ & 0 & 1 & 0 & 1 & 2 & 0 & 2 & 0 & 1 & 0 & 1 & 0 & 1 & 0 & 1 & 0  \\ 
	\hline
	\hline
	$n$ & 16 & 17 & 18 & 19 & 20 & 21 & 22 & 23 & 24 & 25 & 26 & 27 & 28 & 29 & 30 & 31\\
	\hline
	$g(\Bn)$ & 2 & 0 & 1 & 0 & 2 & 0 & 1 & 0 & 2 & 0 & 1 & 0 & 2 & 0 & 1 & 0 \\ 
	\hline	
	\end{tabular}
\end{center}

\medskip

By Theorem~\ref{th:aperiodic}, we know that this sequence is aperiodic.
By Theorem~\ref{th:NP-IMARK-1-2}, we know that $g(\Bn)=0$ 
for every odd $n$, $n\ge 5$, and $g(\Bn)\in\{1,2\}$ for every even $n$, $n\ge 4$.
Hence, this sequence is 1-almost periodic. We now prove that deciding whether
$g(\Bn)=2$ for any position $n$ is easy.

For every integer $n>0$, we denote by $R^1(n)$ the binary number obtained
from the binary representation $R(n)$ of $n$ by deleting all the trailing 0's.
Then we have:

\begin{theorem}
Let $\Bn$ be any position of the game i-$\MARK(\{1\},\{2\})$. 
Then $g(\Bn)=2$ if and only if $n$ is even and either:
\ESOK{\begin{itemize}
\item[{\rm (i)}] $R^1(n)=11_2$ and $R(n)$ has an odd number of trailing 0's, or
\item[{\rm (ii)}] $R^1(n)\neq 11_2$ and $R(n)$ has an even number of trailing 0's.
\end{itemize}}
\label{th:G-IMARK-1-2-compute}
\end{theorem}

\begin{proof}
The theorem clearly holds for $n\le 6$.
Suppose now that the theorem holds up to position $\B{n-1}$, $n\ge 7$,
and consider the position $\Bn$.

If $n$ is odd, we know by Theorem~\ref{th:NP-IMARK-1-2} that $n\in\PP$ and thus $g(\Bn)=0$.

Assume thus that $n$ is even, so that $\opt(\Bn)=\{\B{n/2},\B{n-1}\}$.
Since $g(\B{n-1})=0$, we get $g(\Bn)=2$ if and only if $g(\B{n/2})=1$ and the result directly
follows from the induction hypothesis.
\end{proof}

Let $\NN_2$ denote the set of integers $n$ such that $g(\Bn)=2$. The set $\NN_2$ is as follows:
$$\NN_2=\{4,6,16,20,24,28,36,44,52,60,64,68,76,80,84,92,96,100,108,112,\ldots\}.$$

Note that from Theorems~\ref{th:NP-IMARK-1-2} and~\ref{th:G-IMARK-1-2-compute}, we get
that computing the $g$-value of any position $\Bn$ for the game $\IMARK(\{1\},\{2\})$
can be done in time $O(\log n)$.

\bigskip

\section{The game \IMARK$([1,t-1],\{d\})$}
\label{sec:imark-t-d}

We consider in this section the 
game $\IMARK([1,t-1],\{d\})$, with $t\ge 2$ and $d\ge 2$.
If $t=d=2$, we get the game $\IMARK(\{1\},\{2\})$ considered in the previous
section. More generally, if $t=d$, we get the
integral version of the game \MARKT\
introduced by Aviezri Fraenkel in~\cite{F12} and studied by Alan Guo in~\cite{G12}.

We will prove that for every $t\ge 2$ and $d\ge 2$, $d\not\equiv 1\pmod t$, 
the outcome sequence of the game $\IMARK([1,t-1],\{d\})$ is periodic, while
its Sprague-Grundy sequence is 1-almost periodic.


We first consider the 
outcome sequence of the game $\IMARK([1,t-1],\{d\})$ when $d\not\equiv 1\pmod t$, and prove the following:

\begin{theorem}
For every integers $t\ge 2$ and $d\ge 2$, $d\not\equiv 1\pmod t$, the 
outcome sequence of the game i-\MARK$([1,t-1],\{d\})$
is periodic, with preperiod of length $d(t-1)+2$ and period of length~$t$.
More precisely,
the set $\PP$ of P-positions is given by
$$\PP=\{\B{qt}\ |\ 0\le q<d\}\cup\{\B{qt+1}\ |\ q\ge d\}.$$
\label{th:NP-iMARK-t-d}
\end{theorem}

\begin{proof} 
We clearly have
$\B{0}\in\PP$, which implies  $\{\B{1},\ldots,\B{t-1}\}\subseteq\NN$
since $\B{0}$ is an option of all these positions.

Suppose now that the theorem holds up to position $\B{n-1}$, $n\ge t$,
and consider the position $\Bn$. Let  $n=qt+r$, with $q>0$ and $r\in [0,t-1]$.

If $q<d$ and $r\neq 0$ then $\B{qt}\in\opt(\Bn)$ and thus, 
since $\B{qt}\in\PP$ by induction hypothesis, $\Bn\in\NN$.

If $q<d$ and $r=0$, $\opt(\Bn)=\{\B{(qt-t+1},\ldots,\B{qt-1}\}$ and thus,
since all these options are in $\NN$ by induction hypothesis, $\Bn\in\PP$.

If $q=d$ and $r=0$ -- that is, $n=dt$ -- then, since $\B{t}\in\opt(\Bn)$ and $\B{t}\in\PP$ by induction hypothesis, $\Bn\in\NN$.

If $q>d$ and $r=0$ then, since $\B{(q-1)t+1}\in\opt(\Bn)$ and $\B{(q-1)t+1}\in\PP$ by induction hypothesis, $\Bn\in\NN$.

If $q\ge d$ and $r=1$ we consider two cases.
If $d\hspace{-3pt}\not |n$ then $\opt(\Bn)=\{\B{(q-1)t+2}$,$\ldots$, $\B{qt}\}$ and thus, since all these options are in $\NN$
by induction hypothesis, $\Bn\in\PP$.
If $d|n$ then $\opt(\Bn)=\{\B{n/d}\}\cup\{\B{(q-1)t+2},\ldots,\B{qt}\}$.
Since $d\not\equiv 1\pmod t$ and $n\equiv 1\pmod t$ we have $n/d\not\equiv 1\pmod t$.
Hence, all the options of $\Bn$ are in $\NN$
by induction hypothesis, and thus $\Bn\in\PP$.

Finally, if $q\ge d$ and $r>1$ then, since $\B{qt+1}\in\opt(\Bn)$ and $\B{qt+1}\in\PP$ by induction hypothesis, $\Bn\in\NN$.

This gives that the outcome sequence of the game $\IMARK([1,t-1],\{d\})$
is periodic, with preperiod of length $dt-1$ and period of length~$t$.
\end{proof}

Using computer check, it seems that, for every integers $t\ge 2$ and $d\ge 2$, $d\not\equiv 1\pmod t$,
the Sprague-Grundy sequence of the game 
$\IMARK([1,t-1],\{d\})$ is 1-almost periodic with period of length $t$. 
However, the length of the preperiod seems to be more ``erratic''.

We thus prove this property only for two particular cases, namely $d=t$ and $d=2$.
Moreover, we will also prove that, in these two cases, the non-periodic column has a nice structure,
so that the $g$-value of any position $\Bn$ can be computed in time $O(\log n)$.

\subsection{The game $i$-\MARK$([1,t-1],\{t\})$}
\label{subsec:imark-t}

We now prove that 
the Sprague-Grundy sequence of the game $\IMARK([1,t-1],\{t\})$ is 1-almost periodic for every $t\ge 2$
and that the $g$-value of any position $\Bn$ can be computed in time $O(\log n)$.

\begin{theorem}
For every $t\ge 2$, the Sprague-Grundy sequence of the game i-\MARK$([1,$ $t-1],\{t\})$ is 1-almost periodic, with preperiod of length $t^2-1$, period of length $t$
and exception set $\EE=\{0\}$.
More precisely,
for every integer $n=qt+r$ with $q\ge 0$ and $r\in[0,t-1]$, we have
\begin{itemize}
\item[{\rm (i)}] if $q<t$ and $r=0$, or $q\ge t$ and $r=1$, then $g(\Bn)=0$,
\item[{\rm (ii)}] if $q<t$ and $r\neq 0$, then $g(\Bn)=r$,
\item[{\rm (iii)}] if $q\ge t$ and $r>1$, then $g(\Bn)=r-1$,
\item[{\rm (iv)}] if $q\ge t$ and $r=0$, then $g(\Bn)\in\{t-1,t\}$.
\end{itemize}
\label{th:G-iMARKt}
\end{theorem}

\begin{proof} 
Claim (i) direclty follows from Theorem~\ref{th:NP-iMARK-t-d}.
For the remaining claims, the proof easily follows by induction.
Suppose that the theorem holds up to position $\B{n-1}$, $n\ge 1$,
and consider the position $\Bn$. Let $n=qt+r>1$, with $q\ge 0$ and $r\in[0,t-1]$.

If $q=0$ and $r\neq 0$ then $\opt(\Bn)=\{\B{0},\ldots,\B{r-1}\}$ and the result follows thanks to the induction hypothesis.

Similarly, if $q>0$ and $r\neq 0$ then $\opt(\Bn)=\{\B{qt+r-t+1},\ldots,\B{qt+r-1}\}$ and the result follows thanks to the induction hypothesis.

Finally, if $q\ge t$ and $r=0$ then $\opt(\Bn)=\{\B{q}\}\cup\{\B{qt-t+1},\ldots,\B{qt-1}\}$ and, again, the result follows thanks to the induction hypothesis.

Hence,
the Sprague-Grundy function $\GG($i-\MARK-t$)$ is 1-almost periodic, with preperiod of length $t^2-1$, period of length $t$
and exception set $\EE=\{0\}$.
\end{proof}

As quoted in Section~\ref{sec:intro}, Alan Guo obtained in~\cite{G12} a result with a similar flavour for the game \MARKT.
He also proved that the $g$-value of any position $\Bn$ can be computed
in time $O(n^2)$, using the representation of $n$ in base $t$. 

Our next result, combined with Theorem~\ref{th:G-iMARKt}, shows that the situation for the game $\IMARK([1,t-1],\{t\})$
is easier since the $g$-value of any position $\Bn$ can be computed
in time $O(\log n)$.

For every integer $n\ge 0$, let us denote by $R_t(n)$ the representation of $n$
in base $t$ and by
$R_t^1(n)$  the number in base $t$ obtained
from $R_t(n)$ by deleting all the trailing 0's.
We then have:

\begin{theorem}
Let $t\ge 2$ be an  integer.
Let $\Bn$ be any position of the game i-$\MARK([1,t-1],\{t\})$ with $n=qt$, $q\ge t$.
Then $g(\Bn)=t$ if and only if either:
\begin{itemize}
\item[{\rm (i)}] $R_t^1(n)\in\{R_t^1(k)\ |\ k<t^2,\ k\equiv t-1\pmod t\}$ and $R_t(q)$ has an odd number of trailing 0's, or
\item[{\rm (ii)}] $R_t^1(n)\notin\{R_t^1(k)\ |\ k<t^2,\ k\equiv t-1\pmod t\}$ and $R_t(q)$ has an even number of trailing 0's.
\end{itemize}
\label{th:G-iMARKt-compute}
\end{theorem}

\begin{proof} 
From Theorem~\ref{th:G-iMARKt}, 
we know that for every position $\Bn$, with $n=qt$ and $q\ge t$, $g(\Bn)=t$
if and only if $g(\B{q})=t-1$, which happens only if 
either $q=q't-1$ with $q'\le t$, or $q=q't$ with $q'\ge t$ and $g(\B{q'})=t$.

The result directly follows from this observation.
\end{proof}


\subsection{The game $i$-\MARK$([1,t-1],\{2\})$}
\label{subsec:imark-t-2}

For the game $\IMARK([1,t-1],\{2\})$ we will prove 
that the Sprague-Grundy sequence is 1-almost periodic for every $t\ge 3$
and that the $g$-value of any position $\Bn$ can be computed in time $O(\log n)$.

Note that the case $t=2$ corresponds to the game $\IMARK(\{1\},\{2\})$
considered in Subsection~\ref{subsec:imark-12}.
We will consider separately the case $t=3$, and then the general case $t\ge 4$.

\subsubsection{The game $i$-\MARK$([1,2],\{2\})$}
\label{sub:imark-12-2}

We first consider the case $t=3$ and prove that the Sprague-Grundy sequence
of the game $\IMARK([1,2],\{2\})$ 
is 1-almost periodic.

\begin{theorem}
The Sprague-Grundy sequence of the game i-\MARK$([1,2],\{2\})$
is 1-almost periodic, with preperiod of length 18,  period of length 3, and exception set $\EE=\{0\}$.
More precisely,
for every integer $n=3q+r$ with $q\ge 6$ and $r\in[0,2]$, we have
\begin{itemize}
\item[{\rm (i)}] if $r=1$ then $g(\Bn)=0$,
\item[{\rm (ii)}] if $r=2$ then $g(\Bn)=1$,
\item[{\rm (iii)}] if $r=0$ then $g(\Bn)\in\{2,3\}$.
\end{itemize}
\label{th:G-iMARK3-2}
\end{theorem}

\begin{proof} 
We will determine the value of $g(\Bn)$ for every position $n$, $n\ge 0$, from which
the theorem will follow.

An easy calculation gives the following first values of $\GG(\IMARK$([1,2],\{2\})$)$:

\medskip

\begin{center}
  	\begin{tabular}{|c|ccccccccccc|}
	\hline
	$n$ & 0 & 1 & 2 & 3 & 4 & 5 & 6 & 7 & 8 & 9 & 10  \\
	\hline
	$g(\Bn)$ & 0 & 1 & 2 & 0 & 1 & 2 & 3 & 0 & 2 & 1 & 0 \\ 
	\hline
	\hline
	$n$ & 11 & 12 & 13 & 14 & 15 & 16 & 17 & 18 & 19 & 20 & 21  \\
	\hline
	$g(\Bn)$ & 2 & 1 & 0 & 2 & 1 & 0 & 2 & 3 & 0  & 1 & 2 \\ 
	\hline
	\end{tabular}
\end{center}

\medskip

In particular, $g(\B{18})=3$, $g(\B{19})=0$ and $g(\B{20})=1$, so that the claimed
property holds for $q=6$ and $r\in[0,2]$.
Suppose now that the theorem holds up to $n-1$, $n\ge 21$, and consider the position $\Bn$.
Let  $n=3q+r$, with $q\ge 7$ and $r\in[0,2]$.

If $r=1$, the options of $\Bn$ are $\B{3(q-1)+2}$, $\B{3q}$, and $\B{(3q+1)/2}$ if $q$ is odd.
Since $(3q+1)/2\equiv 2\pmod 3$ if $q$ is odd, these three potential options are all N-positions 
thanks to the induction hypothesis, which gives $g(\Bn)=0$.

If $r=2$, the options of $\Bn$ are $\B{3q}$, $\B{3q+1}$, and $\B{(3q+2)/2}$ if $q$ is even.
Since $(3q+2)/2\equiv 1\pmod 3$ if $q$ is even, the $g$-value of any of these three potential options is either 0, 2 or 3,
thanks to the induction hypothesis, which gives $g(\Bn)=1$.

Finally, if $r=0$, the options of $\Bn$ are $\B{3(q-1)+1}$, $\B{3(q-1)+2}$, and $\B{3q/2}$ if $q$ is even.
Thanks to the induction hypothesis, we have $g(\B{3(q-1)+1})=0$ and $g(\B{3(q-1)+2})=1$.
Moreover, since $3q/2\equiv 0\pmod 3$ if $q$ is even, we have $g(\B{3q/2})\in\{2,3\}$ and, therefore,
$g(\Bn)=2$ if $q$ is odd, or $q$ is even and $g(\B{3q/2})=3$, and $g(\Bn)=3$ otherwise.

Hence, the Sprague-Grundy sequence of the game i-\MARK$([1,2],\{2\})$
is 1-almost periodic, with preperiod of length 18,  period of length 3, and exception set $\EE=\{0\}$.
\end{proof}

From the proof of Theorem~\ref{th:G-iMARK3-2}, we can see that the set of integers $n$, $n\ge 18$,
for which $g(\Bn)=3$ has somehow a nice structure. This fact is stated in a more explicit
way in the following theorem:

\begin{theorem}
Let $\Bn$ be any position of the game i-$\MARK([1,2],\{2\})$ with $n=3q$, $q\ge 6$.
Then $g(\Bn)=3$ if and only if either:
\begin{itemize}
\item[{\rm (i)}] $q$ is even, $R^1(q)=1_2$ or $R^1(q)=101_2$, and $R(q)$ has an 
\ESOK{even} number of trailing 0's, or
\item[{\rm (ii)}] $q$ is even, $R^1(q)\neq 1_2$,  $R^1(q)\neq 101_2$, and $R(q)$ has an 
\ESOK{odd} number of trailing 0's.
\end{itemize}
\label{th:G-iMARK3-2-compute}
\end{theorem}

\begin{proof} 
The proof directly follows from the last part of the proof of Theorem~\ref{th:G-iMARK3-2}, by
observing that $g(\B{24})=2$ (which corresponds to the case $q=8$, $8=1000_2$), 
$g(\B{30})=2$ (which corresponds to the case $q=10$, $10=1010_2$),
while
$g(\B{18})=3$ (which corresponds to the case $q=6$, \ESOK{$6=110_2$}),
and $g(\B{3q})=2$ for every odd $q$, $q\ge 6$.
\end{proof}

\ESOK{Let $\NN_3$ denote the set of integers $n$ such that $g(\Bn)=3$. The set $\NN_3$ is as follows:
$$\NN_2=\{6,18,42,48,54,60,66,72,78,90,102,114,126,138,150,162,168,\ldots\}.$$}

From Theorems~\ref{th:G-iMARK3-2} and~\ref{th:G-iMARK3-2-compute}, we get that
the $g$-value of any position $\Bn$ of the game $\IMARK([1,2],\{2\})$
can be computed in time $O(\log n)$.

\subsubsection{The game $i$-\MARK$([1,t-1],\{2\})$, $t\ge 4$}
\label{sub:imark-t-2}

We now turn to the general case and prove that 
the Sprague-Grundy sequence $\GG(\IMARK([1,t-1],\{2\}))$ is 1-almost periodic for every $t\ge 4$,
as for the game considered
in the previous subsection, except that the $g$-values 1 and 2 are ``switched''.

\begin{theorem}
For every $t\ge 4$, the Sprague-Grundy sequence of the game i-\MARK$([1,$ $t-1],\{2\})$
is 1-almost periodic, with preperiod of length $2t-1$, period of length $t$, and exception set $\EE=\{0\}$.
More precisely,
for every integer $n=qt+r$ with $q\ge 0$ and $r\in[0,t-1]$, we have
\begin{itemize}
\item[{\rm (i)}] if $q<2$ and $r=0$, or $q\ge 2$ and $r=1$, then $g(\Bn)=0$,
\item[{\rm (ii)}] if $q<2$ and $r\neq 0$, then $g(\Bn)=r$,
\item[{\rm (iii)}] if $q\ge 2$ and $r=2$, then $g(\Bn)=2$,
\item[{\rm (iv)}] if $q\ge 2$ and $r=3$, then $g(\Bn)=1$,
\item[{\rm (v)}] if $q\ge 2$ and $r>3$, then $g(\Bn)=r-1$,
\item[{\rm (vi)}] if $q\ge 2$ and $r=0$, then $g(\Bn)\in\{t-1,t\}$.
\end{itemize}
\label{th:G-iMARKt-2}
\end{theorem}

\begin{proof} 
The proof is similar to the proof of Theorem~\ref{th:G-iMARK3-2}.
Claim (i) direclty follows from Theorem~\ref{th:NP-iMARK-t-d} since, in that case, $\Bn$ has an option in $\PP$
and is thus an N-position.
For the remaining claims, the proof easily follows by induction.
Suppose the theorem holds up to position $\B{n-1}$, $n\ge 1$, and
consider the position $\Bn$. Let $n=qt+r>1$, with $q\ge 0$ and $r\in[0,t-1]$.

If $q=0$ and $r\neq 0$ then $\opt(\Bn)=\{\B{0},\ldots,\B{r-1}\}$ and the result follows thanks to the induction hypothesis.

Similarly, if $q>0$ and $r\neq 0$ then $\opt(\Bn)=\{\B{qt+r-t+1},\ldots,\B{qt+r-1}\}$ and the result follows thanks to the induction hypothesis.

Finally, if $q>0$ and $r=0$ then $\opt(\Bn)=\{\B{qt-t+1},\ldots,\B{qt-1}\}$ if $q$ and $t$ are odd, and
$\opt(\Bn)=\{\B{qt/2}\}\cup\{\B{qt-t+1},\ldots,\B{qt-1}\}$ otherwise.
Hence, thanks to the induction hypothesis, we get $g(\Bn)=t-1$ if $q$ and $t$ are odd,
or at least one of them is even and $g(\B{n/2})=t$, and $g(\Bn)=t$ otherwise.

Hence, the Sprague-Grundy sequence of the game i-\MARK$([1,t-1],\{2\})$
is 1-almost periodic, with preperiod of length $2t-1$, period of length $t$, and exception set $\EE=\{0\}$.
\end{proof}

As for the game \IMARK$([1,2],\{2\})$, 
the $g$-value of any position $\Bn$
of the game \IMARK$([1,t-1],\{2\})$ can be computed in time $O(\log n)$,
thanks to Theorem~\ref{th:G-iMARKt-2} and the following characterization,
the proof of which is similar to the proof of Theorem~\ref{th:G-iMARK3-2-compute}
and is thus omitted:

\begin{theorem} 
Let $t\ge 2$ be an  integer. 
Let $\Bn$ be any position of the game i-$\MARK([\ESOK{1,t-1}],\{2\})$ with $n=qt$, $q\ge 2$,
 $m$ be the smallest integer such that $m\ge 2t$ and $R^1(n)=R^1(m)$,
 and $z_n$ (resp. $z_m$) denote the number of trailing 0's of $R(n)$ (resp. of $R(m)$).
Then $g(\Bn)=t$ if and only if either:
\begin{itemize}
\item[{\rm (i)}] $qt$ is even, $g(\B{m})=t-1$ and $z_n-z_m$ is odd, or
\item[{\rm (ii)}] $qt$ is even, $g(\B{m})=t$ and $z_n-z_m$ is even.
\end{itemize}
\label{th:G-iMARKt-2-compute}
\end{theorem}

\bigskip

\section{The game \IMARK$(\{a,2a\},\{2\})$}
\label{sec:imark-2a}

We consider in this section the game $\IMARK(\{a,2a\},\{2\})$, $a\ge 1$.
Note that when $a=1$, this game is the game $\IMARK([1,2],\{2\})$ considered in Subsection~\ref{sub:imark-12-2}.
We will prove that, in some cases, the outcome sequence of the game $\IMARK(\{a,2a\},\{2\})$
is 1-almost periodic.
We will also prove that, in those cases, the $g$-value of any position $\Bn$ can be computed in time $O(\log n)$.

For any even $a$, the $g$-value of any position $\Bn$, $n$ odd, is easy to determine.
For every integer sequence $A=(a_i)_{i\in\NNN}$, we define the {\em odd subsequence} of $A$
as the subsequence $A_{odd}=(a_{2q+1})_{q\in\NNN}$. We then have:

\begin{theorem}
For every integer $a\ge 2$, $a$ even, the odd subsequence of the Sprague-Grundy sequence
of the game $i$-$\MARK(\{a,2a\},\{2\})$ is purely periodic, with period of length $3a$.
More precisely, for every integer $n$  with $n=3qa+r$, $q\ge 0$ and $r\in[1,3a-1]$, $r$ odd, 
we have
\begin{itemize}
\item[{\rm (i)}]  If $r<a$ then $g(\Bn)=0$,
\item[{\rm (ii)}]  If $a<r<2a$ then $g(\Bn)=1$,
\item[{\rm (iii)}]  If \ESOK{$r>2a$} then $g(\Bn)=2$.
\end{itemize}
\label{th:IMARK-a-2a-n-odd}
\end{theorem}

\begin{proof}
If $n<a$ then $\Bn$ has no option and, therefore, $g(\Bn)=0$.
If $a<n<2a$ then $\opt(\Bn)=\{\B{n-a}\}$ and, therefore, $g(\Bn)=1$.
If $2a<n<3a$ then $\opt(\Bn)=\{\B{n-2a},\B{n-a}\}$ and, therefore, $g(\Bn)=2$.
Hence, the theorem holds for every position $\Bn$ with $n$ odd, $n<3a$.
The result then easily follows by induction since, for every odd $n$,
$\opt(\Bn)=\{\B{n-2a},\B{n-a}\}$.
\end{proof}

However, the whole Sprague-Grundy sequence of the game $\IMARK(\{a,2a\},\{2\})$
is not always 1-almost periodic with period of length $3a$.
The smallest length of the period for which this sequence is 1-almost periodic
\ESOK{is} for instance 18 if $a=3$, 60 if \ESOK{$a=5$} and 36 if \ESOK{$a=6$}.

The Sprague-Grundy sequence of the game $\IMARK(\{a,2a\},\{2\})$
has been proved to be 1-almost periodic with period of length $3a$
when $a=1$ in Theorem~\ref{th:G-iMARK3-2}. 
We will consider the cases $a=2$ and $a=4$ in the two following subsections.

%
%
%

\subsection{The game $i$-\MARK$(\{2,4\},\{2\})$}
\label{subsec:imark-24-2}

We first consider the case $a=2$ and prove that the Sprague-Grundy sequence
of the game $\IMARK(\{2,4\},\{2\})$ is 1-almost periodic.

\begin{theorem}
The Sprague-Grundy sequence of the game i-\MARK$(\{2,4\},\{2\})$
is 1-almost periodic, with preperiod of length 11, period of length 6, and exception set $\EE=\{0\}$.
More precisely,
for any integer $n\ge 11$, $n=6q+r$ with $r\in[0,5]$, we have
\begin{itemize}
\item[{\rm (i)}] if $r=1$ or $r=4$ then $g(\Bn)=0$,
\item[{\rm (ii)}] if $r=2$ or $r=3$ then $g(\Bn)=1$,
\item[{\rm (iii)}] if $r=5$  then $g(\Bn)=2$,
\item[{\rm (iii)}] if $r=0$ then $g(\Bn)\in\{2,3\}$.
\end{itemize}
\label{th:G-iMARK-24-2}
\end{theorem}

\begin{proof} 
We will determine the value of $g(\Bn)$ for every position $n$, $n\ge 0$, from which
the theorem will follow.

An easy calculation gives the following first values of $\GG(\IMARK$(\{2,4\},\{2\})$)$:

\medskip

\begin{center}
  	\begin{tabular}{|c|cccccccccccc|}
	\hline
	$n$ & 0 & 1 & 2 & 3 & 4 & 5 & 6 & 7 & 8 & 9 & 10 & 11 \\
	\hline
	$g(\Bn)$ & 0 & 0 & 1 & 1 & 2 & 2 & 0 & 0 & 1 & 1 & 3 & 2 \\ 
	\hline
	\hline
	$n$ & 12 & 13 & 14 & 15 & 16 & 17 & 18 & 19 & 20 & 21 & 22 & 23\\
	\hline
	$g(\Bn)$ & 2 & 0 & 1 & 1 & 0 & 2 & 2 & 0 & 1 & 1 & 0 & 2 \\ 
	\hline
	\end{tabular}
\end{center}

\medskip

Hence,  the theorem holds up to $n=23$, that is $q=3$ and $r=5$.
Suppose now that the theorem holds up to $n-1$, $n\ge 23$, and consider the position $\Bn$.
Let  $n=6q+r$, with $q\ge 4$ and $r\in[0,5]$.

If $r$ is odd, the options of $\Bn$ are $\B{n-2}$ and $\B{n-4}$ and the claim follows 
thanks to the induction hypothesis.

If $r$ is even, the options of $\Bn$ are $\B{n-2}$, $\B{n-4}$ and $\B{n/2}$. Again, the claim follows 
thanks to the induction hypothesis.

Hence, the Sprague-Grundy sequence of the game i-\MARK$(\{2,4\},\{2\})$
is 1-almost periodic, with preperiod of length 11, period of length 6, and exception set $\EE=\{0\}$.
\end{proof}

From the proof of Theorem~\ref{th:G-iMARK-24-2}, we can see that the set of integers $n$, $n\ge 11$,
for which $g(\Bn)=3$ is not difficult to characterize:

\begin{theorem}
Let $\Bn$ be any position of the game i-$\MARK(\{2,4\},\{2\})$ with $n=6q$, $q\ge 2$.
Then $g(\Bn)=3$ if and only if either:
\begin{itemize}
\item[{\rm (i)}] $q$ is even, $q\ge 4$, $R^1(q)=1_2$ and $R(q)$ has an even number of trailing 0's, or
\item[{\rm (ii)}] $q$ is even, $q\ge 4$, $R^1(q)\neq 1_2$, and $R(q)$ has an odd number of trailing 0's.
\end{itemize}
\label{th:G-iMARK-24-2-compute}
\end{theorem}

\begin{proof} 
The theorem clearly holds for $q=2$ and $q=3$ since $g(\B{12})=g(\B{18})=2$.
Assume the theorem holds up to $q-1$, $q\ge 4$ and consider the position $\Bn$, $n=6q$.
We have  $\opt(\Bn)=\{\B{n-2},\B{n-4},\B{n/2}\}$
with, by Theorem~\ref{th:G-iMARK-24-2}, $g(\B{n-2})=0$ and $g(\B{n-4})=1$.
Therefore, $g(\Bn)=3$ if and only if $g(\B{n/2})=2$.
This happens only if $q$ is even
 -- since, if $q$ is odd then $n/2\equiv 3\pmod 6$ and
$g(\B{n/2})=1$ by Theorem~\ref{th:G-iMARK-24-2} --
and, thanks to the induction hypothesis,
if either 
$R^1(q)=1_2$ and $R(q)$ has an even number of trailing 0's, or
 $R^1(q)\neq 1_2$ and $R(q)$ has an odd number of trailing 0's.
\end{proof}

From Theorems~\ref{th:G-iMARK-24-2} and~\ref{th:G-iMARK-24-2-compute}, we get that
the $g$-value of any position $\Bn$ of the game $\IMARK(\{2,4\},\{2\})$
can be computed in time $O(\log n)$.

\subsection{The game $i$-\MARK$(\{4,8\},\{2\})$}
\label{subsec:imark-48-2}

We now consider the case $a=4$.
The results are similar to the previous ones and we omit the proofs, which use
the same technique.

\begin{theorem}
The Sprague-Grundy sequence of the game i-\MARK$(\{4,8\},\{2\})$
is 1-almost periodic, with preperiod of length 17, period of length 12, and exception set $\EE=\{0\}$.
More precisely,
for any integer $n\ge 17$, $n=12q+r$ with $r\in[0,11]$, we have
\begin{itemize}
\item[{\rm (i)}] if $r\in\{1,3,4,10\}$ then $g(\Bn)=0$,
\item[{\rm (ii)}] if $5\le r\le 8$ then $g(\Bn)=1$,
\item[{\rm (iii)}] if $r\in\{2,9,11\}$  then $g(\Bn)=2$,
\item[{\rm (iii)}] if $r=0$ then $g(\Bn)\in\{2,3\}$.
\end{itemize}
\label{th:G-iMARK-48-2}
\end{theorem}

We just give the  first values of $\GG(\IMARK$(\{4,8\},\{2\})$)$,
including the preperiod and the period:

\medskip

\begin{center}
  	\begin{tabular}{|c|cccccccccccc|}
	\hline
	$n$ & 0 & 1 & 2 & 3 & 4 & 5 & 6 & 7 & 8 & 9 & 10 & 11 \\
	\hline
	$g(\Bn)$ & 0 & 0 & 1 & 0 & 2 & 1 & 2 & 1 & 1 & 2 & 0 & 2 \\ 
	\hline
	\hline
	$n$ & 12 & 13 & 14 & 15 & 16 & 17 & 18 & 19 & 20 & 21 & 22 & 23\\
	\hline
	$g(\Bn)$ & 0 & 0 & 3 & 0 & 2 & 1 & 1 & 1 & 1 & 2 & 0 & 2 \\ 
	\hline
	\hline
	$n$ & 24 & 25 & 26 & 27 & 28 & 29 & 30 & 31 & 32 & 33 & 34 & 35 \\
	\hline
	$g(\Bn)$ & 3 & 0 & 2 & 0 & 0 & 1 & 1 & 1 & 1 & 2 & 0 & 2 \\ 
	\hline
	\hline
	$n$ & 36 & 37 & 38 & 39 & 40 & 41 & 42 & 43 & 44 & 45 & 46 & 47\\
	\hline
	$g(\Bn)$ & 3 & 0 & 2 & 0 & 0 & 1 & 1 & 1 & 1 & 2 & 0 & 2 \\ 
	\hline
	\end{tabular}
\end{center}

\medskip

The set of integers $n$, $n\ge 17$,
for which $g(\Bn)=3$ is characterized as follows:

\begin{theorem}
Let $\Bn$ be any position of the game i-$\MARK(\{4,8\},\{2\})$ with $n=12q$, $q\ge 2$.
Then $g(\Bn)=3$ if and only if $R(q)$ has an odd number of trailing 0's.
\label{th:G-iMARK-48-2-compute}
\end{theorem}

From Theorems~\ref{th:G-iMARK-48-2} and~\ref{th:G-iMARK-48-2-compute}, we get that
the $g$-value of any position $\Bn$ of the game $\IMARK(\{4,8\},\{2\})$
can be computed in time $O(\log n)$.

\bigskip

\section{\IMIMARK: \IMARK\ under mis\`ere convention}
\label{sec:misere}

We consider in this section the game \IMARK\ under mis\`ere convention -- that is, when the first
player unable to move wins the game --, called \IMIMARK. The position $\B{0}$ is thus a P-position
for the game \IMIMARK.

In~\cite{F12}, Aviezri Fraenkel characterized the sets of N- and P-positions of the game \MIMARK\ -- \MARK\ under
mis\`ere convention -- and, as for the normal convention, these two sets appeared to be aperiodic.

We will show that the outcome sequence of the game $\IMIMARK(S,D)$ is {\em purely} periodic -- that is, with no preperiod --
 in many cases, in particular for most of the cases considered in the previous sections.

\subsection{The game $i$-\MIMARK$([1,t-1],D)$}
\label{subsec:imimark-t}

We first consider the case $S=[1,t-1]$ with $t\ge 2$. We first prove the following general result:

\begin{theorem}
For every integer $t\ge 2$ and every set $D\in\NNN_{\ge 2}$
such that $d\not\equiv 1\pmod t$ for every $d\in D$, 
the outcome sequence of the game i-\MIMARK$([1,t-1],D)$
is purely periodic with period of length $t$.
More precisely, the set $\PP$ of P-positions 
is given by
$$\PP=\{\B{qt+1}\ |\ q\ge 0\}.$$
\label{th:NP-iMIMARKt}
\end{theorem}

\begin{proof} 
We clearly have $\B{0}\in\NN$ and $\B{1}\in\PP$.
Hence, for every $n\in[2,t]$, $\Bn\in\NN$ since, in that case, $\B{1}\in\opt(\Bn)$.
Consider now the position $\Bn$, $n=qt+r$, $q\ge 1$ and $r\in[0,t-1]$, and assume that the theorem is true up
to position $\B{n-1}$.

If $r=0$, then $\B{(q-1)t+1}\in\opt(\Bn)$ and thus $\Bn\in\NN$ since  $\B{(q-1)t+1}\in\PP$ by induction hypothesis.
Similarly, if $r>1$ then $\B{qt+1}\in\opt(\Bn)$ and thus $\Bn\in\NN$ since  $\B{qt+1}\in\PP$ by induction hypothesis.

Finally, if  $r=1$ then 
$$\opt(\Bn)=\{\B{qt-t+2},\ldots,\B{qt}\}\cup\{\B{(qt+1)/d}\ |\ d\in D,\ d|n\}.$$
Since $d\not\equiv 1\pmod t$ for every $d\in D$, $(qt+1)/d\not\equiv 1\pmod t$ and, therefore,  $\Bn\in\PP$ thanks to induction hypothesis.

Hence, the outcome sequence of the game \IMIMARK$([1,t-1],D)$
is purely periodic with period of length $t$.
\end{proof}

Note  that for every integer $t\ge 3$ the game $\IMIMARK([1,t-1],\{t\})$ -- the mis\`ere version of the game considered in
Subection~\ref{subsec:imark-t}, or in Subsection~\ref{subsec:imark-t-2} if $t=2$ -- satisfies the hypothesis of 
Theorem~\ref{th:NP-iMIMARKt}.


\subsection{The game $i$-\MIMARK$(\{a,2a\},\{2\})$}
\label{subsec:imimark-2a}

We now consider the case $S=\{a,2a\}$, $a\ge 1$, and $D=\{2\}$.
We will prove that the sets $\NN$ and $\PP$ are purely periodic, with period of
length $3a$, whenever $a=2$ or $a$ is odd.

This claim is easy to prove when $a=2$:

\begin{lemma}
Let $\Bn$ be any position of the game i-\MIMARK$(\{2,4\},\{2\})$.
We then have $\Bn\in\PP$ if and only if $n\equiv 2\pmod 6$ or $n\equiv 3\pmod 6$.
\label{lem:a2}
\end{lemma}

\begin{proof} 
We prove this result by induction.
We clearly have $\B{0},\B{1}\in\NN$ and $\B{2}\in\PP$.
Suppose now that the theorem holds up to $n-1$, $n\ge 4$, and consider the position $\Bn$.
If $n\equiv 0,1,4$ or $5\pmod 3$ then $\Bn\in\NN$ since $\B{n-4}$, $\B{n-4}$, $\B{n-2}$ or $\B{n-2}$, respectively,
are P-positions thanks to the induction hypothesis.
If $n\equiv 3\pmod 6$ then $\Bn\in\PP$ since, thanks to the induction hypothesis, $\B{n-2}\in\NN$ and $\B{n-4}\in\NN$.
Finally, if $n\equiv 2\pmod 6$ then $\Bn\in\PP$ since, thanks to the induction hypothesis, $\B{n-2}\in\NN$, $\B{n-4}\in\NN$
and $\B{n/2}\in\NN$.
\end{proof}

We now consider the case of  $a$ odd. 
The following lemma gives the outcome
of any position $\Bn$ with $n<a$.

\begin{lemma}
Let $\Bn$ be any position of the game i-\MIMARK$(\{a,2a\},\{2\})$, $a$ odd, with $n\in[0,a-1]$. We then have
$\Bn\in\NN$ if and only if $n=0$ or $R(n)$ has an even number of trailing 0's.
\label{lem:0toa-1}
\end{lemma}

\begin{proof} 
Clearly, the result holds for $n=0$.
Assume now that $n>0$. We use induction on the number $z(n)$ of trailing 0's of $R(n)$.
If $z(n)=0$ then $n$ is odd, so that $\Bn$ has no option, which implies $\Bn\in\NN$.
Suppose that the result holds up to $z(n)=k\ge 0$ and let $n$ be such that $z(n)=k+1$.
Since $n$ is even, $\opt(\Bn)=\{\B{n/2}\}$, so that $\Bn$ is an N-position if~and only if $\B{n/2}$
is a P-position and the result follows thanks to the induction hypothesis.
\end{proof}

Incidently, note that Lemma~\ref{lem:0toa-1} also holds when $a$ is even.
We are now able to prove the main result of this subsection: 

\begin{theorem}
Let $a$ be an integer, $a\ge 1$.
If $a=2$ or $a$ is odd then
the outcome sequence of the game i-\MIMARK$(\{a,2a\},\{2\})$ is purely periodic with period of length~$3a$.
\label{th:NP-iMIMARKt-a-2a}
\end{theorem}

\begin{proof} 
The case $a=2$ directly follows from Lemma~\ref{lem:a2}.
Suppose now that $a$ is odd.
We will prove by induction that, for every $n\ge 0$, positions $\Bn$ and $\B{n+3a}$ have the same outcome.
 
Since $\B{0},\B{1}\in\NN$ and $\opt(\B{a})=\{\B{0},\B{1}\}$, we have $\B{a}\in\PP$.
Therefore, $\B{3a}\in\NN$ since $\B{a}\in\opt(\B{3a})$, and thus the property holds for $n=0$.

Suppose now that the theorem holds up to $n-1$, $n\ge 1$, and consider the position~$\Bn$.
We consider two cases, according to the outcome of $\Bn$.

\begin{enumerate}

\item $\Bn\in\PP$.\\
In that case, we have $\B{n+a},\B{n+2a}\in\NN$.
If $n$ is even then $n+3a$ is odd and thus $\opt(\B{n+3a})=\{\B{n+a},\B{n+2a}\}$. Hence, $\B{n+3a}\in\PP$ and we are done.

Suppose now that $n$ is odd and, to the contrary, that $\B{n+3a}\in\NN$.
Note first that since $\Bn\in\PP$ we necessarily have $n\ge a$ by Lemma~\ref{lem:0toa-1}.
Since $\opt(\B{n+3a})=\{\B{n+a},\B{n+2a},\B{(n+3a)/2}\}$ 
and $\B{n+a},\B{n+2a}\in\NN$, we necessarily have $\B{(n+3a)/2}\in\PP$,
which implies $\B{(n+a)/2}\in\NN$ and $\B{(n-a)/2}\in\NN$.
Since $\Bn\in\PP$, we also have $\B{n-a}\in\NN$ which implies
that at least one option of $\B{n-a}$ is in $\PP$.
The possible options of $\B{n-a}$ are
(i)~$\B{(n-a)/2}$, but $\B{(n-a)/2}\in\NN$ as observed above,
(ii)~$\B{n-2a}$ if $n\ge 2a$, but in that case $\B{n-2a}\in\NN$ since $\Bn\in\PP$, and
(iii)~$\B{n-3a}$ if $n\ge 3a$, which therefore necessarily exists and must be in $\PP$.
This implies $\B{(n-3a)/2}\in\NN$ and we finally get 
$$\B{(n-3a)/2}\in\NN,\ \B{(n-a)/2}\in\NN\ \mbox{and}\ \B{(n+a)/2}\in\NN,$$
which contradicts the induction hypothesis.

\item $\Bn\in\NN$.\\
If $\B{n+a}\in\PP$ or $\B{n+2a}\in\PP$ then $\B{n+3a}\in\NN$ and we are done.

Suppose therefore that $\Bn,\B{n+a},\B{n+2a}\in\NN$. 
By induction hypothesis, this implies in particular $\B{n-a}\in\NN$ if $n\ge a$.
If $n$ is even then $n+a$ is odd, so that $\opt(\B{n+a})=\{\B{n-a},\Bn\}$ (resp. $\opt(\B{n+a})=\{\Bn\}$) if $n\ge a$ (resp. $n<a$) and thus $\B{n+a}\in\PP$, contradicting our assumption.
Similarly, if $n$ is odd then $n+2a$ is odd, so that $\opt(\B{n+2a})=\{\B{n},\B{n+a}\}$ and thus $\B{n+2a}\in\PP$, again contradicting our assumption.

\end{enumerate}
This concludes the proof.
\end{proof}


Table~\ref{tab:periods} shows the period, of length $3a$, of the outcome sequence of
the game \IMIMARK$(\{a,2a\},\{2\})$ for $a=2$ or $a$ odd, $a\le 11$.

The period of the outcome sequence of the game $\IMIMARK(\{2,4\},\{2\})$ is given by Lemma~\ref{lem:a2}.
The next two propositions will determine the outcome of any position~$\Bn$, with $a\le n\le 3a-1$,
for the game $\IMIMARK(\{a,2a\},\{2\})$, $a$ odd.
These two propositions, together with Lemma~\ref{lem:0toa-1}, thus determine the corresponding outcome sequences.

\begin{table}
\begin{center}
  	\begin{tabular}{|ccl|}
	\hline
	$a$ & $3a$ & period \\
	\hline 
	1 & 3 & $NPN$\\
	2 & 6 & $NNPPNN$\\
	3 & 9 & $NNPPNNNPN$\\
	5 & 15 & $NNPNNPPNPP\ NNNNN$ \\
	7 & 21 & $NNPNNNPPPN\ PPNNNNNNNP\ N$ \\
	9 & 27 & $NNPNNNPNPP\ PNNPPNNNNN\ NPNNNPN$\\
	11 & 33 & $NNPNNNPNPN\ PPNNPPNNPN\ NNNPNNNPNN\ NPN$ \\
	\hline
	\end{tabular}
\end{center} 
\medskip
\caption{Period of the outcome sequence of the game $\IMIMARK(\{a,2a\},\{2\})$}
\label{tab:periods}
\end{table}

\begin{proposition}
Let $\Bn$ be any position of the game i-\MIMARK$(\{a,2a\},\{2\})$, $a$~odd, with $n\in[a,2a-1]$. We then have:
\begin{itemize}
\item[{\rm (i)}] if $n$ is even then $\Bn\in\PP$,
\item[{\rm (ii)}] if $n$ is odd then $\Bn\in\PP$ if and only if $a=n$ or
$R(n-a)$ has an even number of trailing 0's.
\end{itemize}
\label{prop:ato2a-1}
\end{proposition}

\begin{proof} 
Suppose first that $n$ is even. This implies that $n-a$ is odd, so that 
$\B{n-a}\in\NN$ by Lemma~\ref{lem:0toa-1}.
Since $\B{n-a}\in\opt(\Bn)$, we then get $\Bn\in\PP$. 

Suppose now that $n$ is odd.
If $n=a$ then $\opt(\Bn)=\{\B{0},\B{1}\}$ which gives $\Bn\in\PP$ since $\B{0},\B{1}\in\NN$ by 
Lemma~\ref{lem:0toa-1}. If $n\neq a$, the result again follows from Lemma~\ref{lem:0toa-1} since
$\opt(\Bn)=\{\B{n-a}\}$ and $\B{n-a}\in\NN$ if and only if $R(n-a)$ has an even number of trailing 0's.
\end{proof}

\begin{proposition}
Let $\Bn$ be any position of the game i-\MIMARK$(\{a,2a\},\{2\})$, $a$~odd, with $n\in[2a,3a-1]$. We then have:
\begin{itemize}
\item[{\rm (i)}] if $n=2a$ or $n$ is odd then $\Bn\in\NN$,
\item[{\rm (ii)}] if $n$ is even, $n\neq 2a$, then $\Bn\in\NN$ if and only if
one of the following conditions holds:
   \begin{itemize}
   \item[{\rm (a)}] $R(n-2a)$ has an odd number of trailing 0's,
   \item[{\rm (b)}] $R(n-a)$ has an even number of trailing 0's,
   \item[{\rm (c)}] $n\equiv 0\pmod 4$,
   \item[{\rm (d)}] $n\equiv 2\pmod 4$ and $R(n/2-a)$ has an even number of trailing 0's.
   \end{itemize}
\end{itemize}
\label{prop:2ato3a-1}
\end{proposition}

\begin{proof} 
Since $\B{a}\in\PP$ by Proposition~\ref{prop:ato2a-1}, we have $\B{2a}\in\NN$.
If $n$ is odd then $n-a$ is even, which implies $\Bn\in\NN$ since $\B{n-a}\in\opt(\Bn)$
and $\B{n-a}\in\PP$ by Proposition~\ref{prop:ato2a-1} (case (i)).

Suppose now that $n$ is even, $n\neq 2a$. We then have $\Bn\in\NN$ if and only if at least one of the
positions $\B{n-2a}$, $\B{n-a}$, $\B{n/2}$ is in $\PP$ (note that $n/2\in[a+1,a+(a-1)/2]$).
By Lemma~\ref{lem:0toa-1}, $\B{n-2a}\in\PP$ if and only if $R(n-2a)$ has an odd number of trailing 0's
(case (ii.a)).
Since $n-a$ is odd, by Proposition~\ref{prop:ato2a-1}, $\B{n-a}\in\PP$ if and only if 
$R(n-a)$ has an even number of trailing 0's (case (ii.b)).
Finally, by Proposition~\ref{prop:ato2a-1}, $\B{n/2}\in\PP$ if and only if
either $n/2$ is even, which means $n\equiv 0\pmod 4$ (case (ii.c)),
or $n/2$ is odd, which means $n\equiv 2\pmod 4$, and
$R(n/2-a)$ has an even number of trailing 0's (case (ii.d)).
\end{proof}

When $a$ is even, $a\ge 4$, the outcome sequence of the game $\IMIMARK(\{a,2a\},\{2\})$
seems to be more ``erratic''.
Using computer check, we observed that the outcome sequence seems to be
always periodic, also with period of length $3a$, but with a preperiod of, at that time,
undetermined length.

%

We show in Table~\ref{tab:aeven}, for the first even values of $a$, 
the observed length of the preperiod and the number of {\em exceptions}
contained in the preperiod -- by exception we mean here a position whose outcome is different
from the outcome to which it would correspond in the period. 
The last exception is thus considered as the last element of the preperiod.

\begin{table}

\begin{center}
  	\begin{tabular}{|r|ccccccc|}
	\hline
	$a$ & 4 & 6 & 8 & 10 & 12 & 14 & 16 \\
	\hline
	length of the preperiod & 7 & 9 & 61 & 193 & 105 & 105 & 313\\
	number of exceptions & 2 & 2 & 6 & 6 & 10 & 8 & 14\\
	\hline
	\hline
	$a$ &  18 & 20 & 22 & 24 & 26 & 28 & 30\\
	\hline
	length of the preperiod & 345 & 397 & 425 & 497 & 1129 & 1217 & 585 \\
	number of exceptions & 10 & 14 & 12 & 26 & 16 & 28 & 18\\
	\hline
	\end{tabular}
\end{center} 
\medskip
\caption{Preperiod of the outcome sequence of the game $\IMIMARK(\{a,2a\},\{2\})$, $a$ even}
\label{tab:aeven}
\end{table}

\bigskip

\section{Discussion}
\label{sec:discussion}

In this paper, we initiated the study of a new family of impartial combinatorial
games -- the integral subtraction division games -- obtained by restricting in a natural way
the availability of division-type moves in subtraction division games.
We proved that in many cases these games have a ``nice behaviour'' since, under normal convention,
their Sprague-Grundy sequence is almost periodic and the $g$-value of any heap of $n$ tokens
can be computed in time $O(\log n)$.
Moreover, we proved that, under mis\`ere convention and again in many cases, the outcome sequence is purely periodic.
 
We finally list below a few open questions related to integral subtraction division games
that could be of interest.

\begin{enumerate}
\item Do there exist sets $S$ and $D$ for which the outcome sequence of the game $\IMARK(S,D)$ under normal convention
is not periodic? 

\item Is it true that, for every $d\not\equiv 1\pmod t$, the Sprague-Grundy sequence of the game $\IMARK([1,t-1],\{d\})$ is
1-almost periodic with period of length $t$?

\item What can be said about games of the form $\IMARK([1,t-1],\{d\})$, when $d\equiv 1\pmod t$,
under normal convention? under mis\`ere convention?

\item Is it true that, for every $a\ge 1$, the Sprague-Grundy sequence of the game $\IMARK(\{a,2a\},\{2\})$ is
1-almost periodic?

\item What can be said about games of the form $\IMARK(\{a,2a,\ldots,ka\},\{2\})$, $k\ge 3$, under normal convention?
under mis\`ere convention?

\item What can be said about games of the form $\IMARK(S,D)$, with $|D|>1$, under normal convention?
under mis\`ere convention?

\item Is it true that for every even integer $a\ge 4$ the outcome sequence of the game
\IMIMARK$(\{a,2a\},\{2\})$ is periodic with period of length $3a$ ?

\item What can be said,  when $a$ is odd, about the outcome sequence of the game \IMIMARK$(\{a,2a\},\{2\})$?

\end{enumerate}


\end{document}